\documentclass{article}

\usepackage{latexsym}
\usepackage{amsmath}
\usepackage{amssymb}
\usepackage{mathtools}
\usepackage{enumerate}
\usepackage{color}
\usepackage{stmaryrd}
\usepackage{tikz}
\usepackage{pgfplots}
\usepackage{times}
\usepackage{fullpage}
\usepackage{hyperref}
\hypersetup{
    bookmarks=true,         
    unicode=false,          
    pdftoolbar=true,        
    pdfmenubar=true,        
    pdffitwindow=false,     
    pdfstartview={FitH},    
    pdftitle={My title},    
    pdfauthor={Author},     
    pdfsubject={Subject},   
    pdfcreator={Creator},   
    pdfproducer={Producer}, 
    pdfkeywords={keyword1} {key2} {key3}, 
    pdfnewwindow=true,      
    colorlinks=true,       
    linkcolor=red,          
    citecolor=blue,        
    filecolor=magenta,      
    urlcolor=cyan           
}

\usepackage[official]{eurosym}

\newcommand{\defined}{\vcentcolon =}
\newcommand{\E}{\mathbb E}
\newcommand{\Var}{\operatorname{Var}}

\newcommand{\R}{\mathbb R}
\newcommand{\M}{\mathcal M}
\newcommand{\N}{\mathbb N}
\newcommand{\sr}[1]{\stackrel{#1}}
\newcommand{\set}[1]{\left\{#1\right\}}
\newcommand{\ind}[1]{\llbracket #1 \rrbracket}
\newcommand{\absolute}{\ll}

\newcommand{\eqn}[1]{\begin{align}#1\end{align}}
\newcommand{\eq}[1]{\begin{align*}#1\end{align*}}

\def\subsubsect#1{\vspace{1ex plus 0.5ex minus 0.5ex}\noindent{\bf\boldmath{#1.}}}

\usepackage{amsthm}
\theoremstyle{plain}
\newtheorem{theorem}{Theorem}

\newtheorem{lemma}[theorem]{Lemma}

\theoremstyle{definition}

\theoremstyle{remark}

\newcommand{\origin}{\theta_{\!\circ}}
\newcommand{\bi}{\bar J}
\newcommand{\cesaro}{Ces\`aro}

\newcommand{\spec}[1]{\Vert #1 \Vert_2}

\newcommand{\norm}[1]{\left\vert #1 \right\vert}
\newcommand{\normal}{\mathcal N}
\renewcommand{\epsilon}{\varepsilon}
\newcommand{\T}{\top}
\newcommand{\F}{\mathcal F}

\newcommand{\iid}{i.i.d.}

\title{Asymptotics of Continuous Bayes for Non-i.i.d.\ Sources}

\author{Tor Lattimore$^1$ and Marcus Hutter$^2$ \\[0.5cm] $^1$ Department of Computing Science, University of Alberta \\ $^2$ Research School of Computer Science, Australian National University}
\date{}

\begin{document}

\maketitle

\begin{abstract}
Clarke and Barron analysed the relative entropy between an \iid{} source and a Bayesian mixture over a continuous class
containing that source. In this paper a comparable result is obtained when the source is permitted to be both non-stationary
and dependent. The main theorem shows that Bayesian methods perform well for both compression and sequence prediction
even in this most general setting with only mild technical assumptions. 
\end{abstract}

\section{Introduction}
We continue the work of Clarke and Barron \cite{BC90} bounding the relative entropy between
probability measures on infinite sequences and a Bayesian mixture over measures in some continuous class $\M$.
Small values of relative entropy have a number of implications. Notably that prediction and compression using the Bayesian mixture
is nearly as good as using the true unknown source. The contribution of this work is to show that Bayesian
methods generalise to the case where data is not sampled identically and independently, a situation that is 
often encountered in practical problems.   
One application is online compression where a sequence of words is observed and should be compressed. If
we assume the words are sampled from some probability distribution, which may be non-stationary and dependent, then 
near-optimal compression is obtained by arithmetic
coding with respect to that distribution. Typically, however, the probability distribution from which the data is sampled 
is unknown. One approach in this case is to code with respect to
the Bayesian mixture over some set of measures believed to contain the truth. Bounding the relative entropy
between the truth and the Bayesian mixture is then equivalent to bounding the compression redundancy due to coding
with respect to the wrong measure.

Another application is discriminative learning where a classifier should predict label data based on observations.
Sometimes it is possible to model the joint distribution of the observations and labels together, but in many
cases the observations are relatively unordered and modelling the conditional distribution of the label given
the observation is easier. Our results show that Bayesian methods for discriminative learning perform well
under only mild technical assumptions.

Suppose $\M = \set{P_\theta : \theta \in \Theta \subseteq \R^d}$ is a set of measures on the space of infinite 
sequences over alphabet $\Omega$. We let $M$ be a Bayesian mixture over $\M$ with respect to some prior and
analyse the relative entropy $D(P_\theta^n\Vert M^n)$ where $P_\theta^n$ and $M^n$ are the distributions on the 
first $n$ observations induced by
$P_\theta$ and $M^n$ respectively. Unlike in \cite{BC90} we permit $P_\theta$ to be non-stationary and dependent, so
$P^n_\theta$ is typically not a product measure.

Our main contribution is a proof under mild technical assumptions that the relative entropy can be bounded
by
\eq{
D(P_{\theta}^n \Vert M^n) \leq \ln{1 \over w(\theta)} + {d \over 2} \ln{n \over 2\pi} + {1 \over 2}{\ln \det \bi_n(\theta)} + o(1)
}
where $w(\theta)$ is the prior density of parameter $\theta$ and $\bi_n(\theta)$ is the mean Fisher information matrix at $\theta$.
If $P_\theta^n$ is a product measure, then $\bi_n(\theta)$ coincides with the usual Fisher information matrix and the result above
is the same as that given in \cite{BC90}.

The main difficulty in generalising the proof in \cite{BC90} to non-stationary dependent sources is
the fact that the mean Fisher information matrix $\bi_n(\theta)$ is now dependent on $n$.
We start with some notation (Section \ref{sec:not}). The main theorems are 
presented (Section \ref{sec:thm}) followed by a variety of applications and
discussion (Section \ref{sec:apps}). The proofs are found in Sections \ref{sec:proof} and \ref{sec:proof-weak}. 
We discuss some of the assumptions and special cases in Sections \ref{sec:equi}-\ref{sec:zero}, including the case
when the information matrix vanishes and a comparison to the known results in zero-dimensional (countable) families.
For discussion and conclusions see Section \ref{sec:conc}. 

\section{Notation}\label{sec:not}

We use $\ln$ for the natural logarithm and $A^{\T}$ for the transpose of matrix $A$. Suppose $A \in \R^{d \times d}$ and $x \in \R^d$.
Then $\norm{x}_2^2 \defined x^{\T} x$ is the standard $2$-norm and $\norm{x}_A^2 \defined x^{\T} A x$ is the norm with respect to $A$. Note
that $\norm{\cdot}_A$ is a norm only if $A$ is positive definite, but occasionally we abuse notation by writing $\norm{\cdot}_A$ even if
$A$ is not positive definite.
We use $\spec{A}$
for the spectral norm of $A$, which for positive definite matrices is the largest eigenvalue of $A$. It is easy to check that
$\norm{x}_2^2 \leq \norm{x}_A^2  \spec{A^{-1}}$. The determinant of $A$ is $\det A$.
The indicator function $\ind{expr}$ is equal to $1$ if $expr$ is true and $0$ otherwise. For function $f:\R^n \to \R$ we write $\partial_i f$
for the partial derivative of $f$ with respect to the $i$th coordinate. Higher order derivatives are denoted 
by $\partial_{i,j} f$ or $\partial_{ i_1 \cdots  i_k}f$. The Gamma function is $\Gamma(x) = \int^\infty_0 t^{x-1} e^{-t} dt$.

\subsubsect{Probability spaces}
Let $(\Omega, \F)$ be a measurable space. Then the product measure space is $(\Omega^n, \F^n)$ where $\F^n$ is the $\sigma$-algebra
generated by the $n$-fold tensor product of $\F$. Let $\F^\infty$ be the cylinder $\sigma$-algebra defined by
\eq{
\F^\infty \defined \sigma\set{\bigcup_{n=1}^\infty \F^n \otimes \set{\Omega^\infty}}.
}
Then $(\Omega^\infty, \F^\infty)$ is a measurable space.
A probability measure $P$ on this space may be thought of as a family of probability measures where $P^n: \F^n \to [0,1]$
is induced by restriction $P^n(A) = P(A \times \Omega^\infty)$.
We think of $P$ as a probability measure on infinite sequences in $\Omega^\infty$ and $P^n$ to be the probability measure
on the first $n$ observations induced by $P$. If $\omega \in \Omega^\infty$, then
$\omega_{1:n} = \omega_1\omega_2 \cdots \omega_n \in \Omega^n$ is the projection of the first $n$ components and 
$\omega_{<n} \defined \omega_{1:n-1}$.

\subsubsect{Parameterised families and Bayes}
Suppose $\Theta \subseteq \R^d$ and that $P_\theta$ is a family of probability measures parameterised by $\theta \in \Theta$.
We assume there exists a measure $\nu$ such that $\nu^n$ is $\sigma$-finite for all $n$ and
the density of $P^n_\theta$ with respect to $\nu^n$ exists for all $n$ and $\theta \in \Theta$, which is denoted by $p^n_\theta$.
Then by the definition of the density we have
\eq{
(\forall A \in \F^n) \qquad P^n_\theta(A) = \int_A p^n_\theta(\omega) d\nu^n(\omega).
}
We denote expectations with respect to $P^n_\theta$ by $\E_\theta$ where $n$
is always understood from the context. Variances are denoted by $\Var_{\theta}$.
The $n$-step mean Fisher information matrix at $\theta \in \Theta$ is denoted by $\bi_n(\theta) \in \R^{d\times d}$ and defined
\eq{
\bi_n(\theta)_{i,j} 
\defined -{1 \over n}\E_{\theta} \left[\partial_{i,j} \ln p_\theta^n\right],
}
which like all other quantities of interest is independent of the reference measure used to define the density.
There are other definitions of the Fisher information matrix, all of which coincide under the weak assumption
that derivatives and expectations can be exchanged \cite[\S18]{Gru07}. We are not aware of any interesting
family of measures for which this assumption is not satisfied.
Let $w:\Theta \to [0,\infty)$ be a prior probability density with respect to the Lebesgue measure, which satisfies
$\int_\Theta w(\theta) d\theta = 1$. Then
the density of the Bayesian mixture $m^n: \Omega^n \to [0,\infty)$ is defined by
\eq{
m^n(\omega_{1:n}) \defined \int_{\Theta} w(\theta) p^n_\theta(\omega_{1:n}) d\theta.
}
The Bayes mixture measure $M^n:\F^n \to [0,1]$ may equivalently be defined by
\eq{
M^n(A) &= \int_A m^n(\omega) d\nu(\omega) 
\equiv \int_A \int_\Theta w(\theta) p^n_\theta(\omega) d\theta d\nu(\omega) 
= \int_\Theta w(\theta) P^n_\theta(A) d\theta.
}
If $w$ and $p^n_\theta$ are continuous at $\theta$, then $P^n_\theta$ is absolutely continuous with respect to $M^n$ and
the relative entropy between $P^n_\theta$ and $M^n$ is defined by
\eq{
D(P^n_\theta \Vert M^n) \defined \E_{\theta} \ln {p^n_\theta \over m^n}.
}
We do not claim that $P_\theta \absolute M$ and indeed this is not generally the case ($\M \equiv \text{Bernoulli measures}$).
The $d$-dimensional (non-degenerate) normal distribution has density
\eq{
\normal(\theta|\mu, \Sigma) = {1 \over \sqrt{(2\pi)^d \det \Sigma}} \exp\left(-{1 \over 2} \norm{x - \mu}_{\Sigma^{-1}}^2\right)
}
where $\mu \in \R^d$ is the mean and $\Sigma$ is the positive definite covariance matrix. 

\section{Main Theorems}\label{sec:thm}

Our main result generalises Theorem 2.3 in \cite{BC90}. The conditions (a), (b) and (d) below are standard regularity conditions 
also made in \cite{BC90}. The conditions (c) and (e)
are regularity conditions on the mean Fisher information matrix, which in this work depends on $n$ as well as $\theta$.
The second result relaxes condition (e) at the cost of a slightly worse bound.

\begin{theorem}\label{thm:main}
Let $\origin \in \Theta$ and $\bi_n \equiv \bi_n(\origin)$. Define a family of functions $\set{f_n}: \Theta \to \R$ by
\eq{
f_n(\theta) \defined {1 \over n} D(P^n_{\origin}\Vert P_\theta^n).
}
Assume that:
\begin{enumerate}[(a)]
\item $w$ is continuous at $\origin$.
\item $f_n(\theta)$ is twice differentiable at $\origin$.
\item Entries in the Hessian of $f_n$ are equicontinuous at $\origin$.
\eq{
\lim_{\theta\to\origin}\;\sup_n \norm{\partial_{i,j}f_n(\theta) - \partial_{i,j}f_n(\origin)} = 0.
}
\item ${\partial_i}f_n(\origin)= 0$ for all $i$. 
\item $\limsup_{n\to\infty} \spec{\bi_n^{-1}} < \infty$.
\end{enumerate}
If $D_n \equiv D(P^n_{\origin} \Vert M^n)$, then
\eq{
\limsup_{n\to\infty} \left(D_n - {d \over 2} \ln{n \over 2\pi} - {1 \over 2} \ln \det \bi_n\right) \leq \ln{1 \over w(\origin)}.
}
\end{theorem}

The three components of the bound can be explained as follows. Assume $P_\theta$ is \iid{}, which implies that $\bi_n = \bi_1$ is
independent of $n$. 
We want to approximate $D_n$ by integrating the Bayes mixture over a region $R_n$ containing $\origin$. 
\eq{
D_n &= \E_{\origin} \ln {p^n_{\origin} \over m^n} 
\leq \E_{\origin} \ln {p^n_{\origin} \over \int_{R_n} w(\theta) p^n_\theta d\theta} 
\approx \E_{\origin} \ln {p^n_{\origin} \over \int_{R_n} w(\origin) p^n_{\origin} d\theta} 
= \ln {1 \over w(\origin)} + \ln{1 \over V(R_n)}
}
where $V(R_n)$ is the volume of $R_n$. The quality of the approximation depends on the choice of $R_n$ with smaller regions
leading to better approximations, but also smaller volumes. A single scalar parameter $\theta$ can usually 
be estimated with an accuracy of about
$n^{-1/2}$ using $n$ samples, which suggests that $p^n_{\theta}$ is approximately constant inside the cube
of width $n^{-1/2}$ and dimension $d$. As $n$ tends to infinity the volume of the cube shrinks to zero and the 
continuity of the prior justifies the approximation $w(\theta) \approx w(\origin)$.
Choosing this as the region leads to $V(R_n) = n^{-d/2}$, which explicates the 
${d \over 2} \ln n$ component
of the bound. The additional term depending on the information matrix is explained by making the above argument more precise.
The insight is that $\bi_1$ is a measure of the expected curvature of $\ln p_\theta^1$ at $\origin$. 
Then large $\bi_1$ implies that $R_n$ must be small for a good approximation. If the data is \iid{}, then the
curvature is independent of $n$ and $V(R_n) \propto n^{-d/2} \det \bi_1^{1/2}$ is unsurprising. For non \iid{} sources the curvature
of $\ln p^n_\theta$ is possibly dependent on $n$ and so appears in the bound in Theorem \ref{thm:main}.
An alternative explanation is found in \cite{BC90}.

By making an additional assumption we also obtain a stronger result in terms of the actual redundancy, rather than the
expected redundancy.

\begin{theorem}\label{thm:no-expect}
Assume conditions (a-e) and additionally that 
\begin{flalign*}
(f) &\quad\qquad\qquad\qquad {1 \over n} \Var_{\origin} \left(\ln {p_{\origin}^n  \over p_\theta^n}\right) &
\end{flalign*}
is twice differentiable at $\origin$ and has equicontinuous second derivatives as in assumption (c) above.
Then there exists a sequence of random variables $r_n:\Omega^\infty$ such that
for all $\omega = \omega_1 \omega_2 \cdots \in \Omega^\infty$  
\eq{
\limsup_{n\to\infty} &\left(\ln {p^n_{\origin}(\omega_{1:n}) \over m^n(\omega_{1:n})} - {d \over 2} \ln {n \over 2\pi} -
{1 \over 2}\ln \det \bi_n - r_n(\omega) \right) 
\leq \ln{1 \over w(\origin)} + \ln 2 + \sqrt{2d} 
}
where $\lim_{n \to\infty} \E_{\origin}|r_n| = 0$.
\end{theorem}

The condition (c) above is subtly different to the one suggested in \cite[\S3]{Hut05} where it is assumed that the average 
information matrix is equicontinuous at $\origin$. The version used here is require for technical reasons and may in fact be
necessary.
The condition (d) is very weak and follows from the standard regularity conditions that permit the exchange of the
derivative and expectation.
The assumption (e) that $\spec{\bi_n^{-1}}$ is uniformly bounded in $n$
is unexpected because if the opposite is true, then $\M$ should in some direction be 
approximately flat near $\origin$, which is precisely 
when we might expect to do even better than Theorem \ref{thm:main} implies.

\begin{theorem}\label{thm:main-weak}
Let $\epsilon > 0$. Under conditions (a-d) above it holds that
\eq{
\limsup_{n\to\infty} \left(D_n - {d \over 2} \ln{n \over 2\pi} - {d \over 2} \ln (\spec{\bi_n} + \epsilon) \right) 
\leq \ln{1 \over w(\origin)}.
}
\end{theorem}
Although the choice of $\epsilon$ is arbitrary, the $o(1)$ term hidden by the $\limsup$ depends on $\epsilon$, which prevents $\epsilon = 0$. Theorem
\ref{thm:main-weak} does not imply that $D_n$ grows sub-logarithmically if $\spec{\bi_n} = 0$. Later we show that
$D_n$ will typically grow logarithmically with $n$ except when $f_n$ is completely flat at $\origin$.
The ${d \over 2} \ln (\spec{\bi_n} + \epsilon)$ term may be replaced by ${1 \over 2} \ln \det A_n$ where $A_n = \bi_n + \epsilon I$ with $I$ the
identity matrix. See the proof for details. 
The proofs of Theorems \ref{thm:main}, \ref{thm:no-expect} 
and \ref{thm:main-weak} are delayed until Sections \ref{sec:proof} and \ref{sec:proof-weak}.

\section{Applications}\label{sec:apps}
The conditions of Theorems \ref{thm:main} and \ref{thm:main-weak} are satisfied for many well-known sources
such as categorical sources, Markov models and some types of discriminative learning. We present
some simple examples, but first give some interpretations of the relative entropy as a measure
of performance. Note that many of the applications considered below have been considered elsewhere in more specialized articles. The
main objective of this section is to show that the new result provides comparable bounds on the performance of Bayes,
even in these special cases.

\subsubsect{Interpretations of relative entropy}
For simplicity suppose $\Omega$ is countable, which allows us to replace the integrals by sums and
think of the density $p^n_\theta$ as a distribution on $\Omega^n$.
The relative entropy between $P^n$ and $M^n$ is a versatile measure
of the performance of $M^n$ when predicting or compressing in place of $P^n$.
If $\omega_{1:t} \in \Omega^t$, then the conditional density of measure $p^t$ is written 
$p^t(\omega_t|\omega_{<t}) = p^t(\omega_{1:t})/p^t(\omega_{<t})$.
An important property is the chain rule, which says that
\eq{
D_n \equiv D(P^n \Vert M^n) = \E_P \sum_{t=1}^n d_t
}
where 
\eq{
d_t(\omega_{<t}) 
&\defined \E_P \left[\ln {p^t(\omega_t|\omega_{<t}) \over m^t(\omega_t|\omega_{<t})} \Big\vert \omega_{<t}\right] 
\equiv \sum_{\omega_t \in \Omega} p^t(\omega_t|\omega_{<t}) \ln {p^t(\omega_t|\omega_{<t}) \over m^t(\omega_t|\omega_{<t})} 
} 
is the
$\F^{t-1}$-measurable random variable that is the relative entropy between the $1$-step predictive distributions of 
$P$ and $M$ given $\omega_{<t}$. Therefore if $d_t(\omega_{<t})$ is small, then $m^t(\cdot|\omega_{<t})$ is
close to $p^t(\cdot|\omega_{<t})$ and so predicting using the Bayesian mixture measure $M$ is nearly as good as using
the unknown $P$.
If $D_n$ grows sub-linearly, then $d_t$ converges to zero in \cesaro\ average,
which implies that $M$ predicts almost as well as the optimal unknown distribution $P$.

More directly, suppose $\omega \in \Omega^\infty$ is sampled from an unknown $P_\theta$ and observed sequentially.
An online compression algorithm that knows $P_\theta$ may use arithmetic coding to produce a code for $\omega_{1:n}$ that has
an expected code-length of at most two bits more than the optimal value. If $M$ is used rather than $P$, then
the expected additional code-length is the relative entropy $D_n$. This means that substituting
the Bayesian mixture into the arithmetic coding algorithm achieves an expected code-length of at most
$D_n + 2$ bits more than the theoretical limit \cite{RL79,WNC87,Mac03}. 

\subsubsect{Categorical sources}
Let $d \geq 1$ and
$\Omega = \set{0,1,2,\cdots,d}$. An \iid{} categorical source is a product measure on the infinite sequences $\Omega^\infty$ so 
that the probability of sampling 
symbol $0 \leq k \leq d$ depends only on $k$ and not the preceding symbols. The space of categorical measures
is parameterised by 
\eq{
\Theta = \set{\theta \in [0,1]^{d} : \norm{\theta}_1 \leq 1}
}
where $\theta_k$ is the probability that the source
parameterised by $\theta$ produces symbol $k$. Define 
$\theta_{0} = 1 - \norm{\theta}_1$ and 
\eq{
p_\theta(\omega_{1:n}) = \prod_{t=1}^n \theta_{\omega_t}.
}
Then since $p_{\theta}$ is a product measure 
we have for all $n$ that $\bi_n = \bi_1$, which may easily be computed
by hand to be
\eq{
(\forall 1 \leq i,j \leq d) \qquad \bi_n(\theta)_{i,j} = \begin{cases}
{1 \over \theta_i} + {1 \over \theta_{0}} & \text{if } i =j \\
{1 \over \theta_{0}} & \text{otherwise.}
\end{cases}
}
Then by Lemma \ref{lem:mat} in the Appendix we compute the determinant to be 
$\det \bi_n(\theta) = \prod_{k=0}^{d} {1 \over \theta_k}$.
We are permitted to choose any prior density provided it is continuous, but a natural choice is
Jeffrey's prior, which when it exists is defined to be proportional to the square root of the determinant
of the information matrix. Jeffrey's prior is a non-informative prior that is parameterisation invariant. 
A nice property is that if $w$ is chosen to be Jeffrey's prior, then the bound on $D_n$ will be independent of $\origin$.
For \iid{} categorical sources Jeffrey's prior is the symmetric Dirichlet with parameter $1 \over 2$.
\eq{
w(\theta) = \Gamma\left({d+1\over2}\right) \pi^{-{d+1 \over 2}} \sqrt{\prod_{k=0}^{d} {1 \over \theta_k}}
\propto \sqrt{\det \bi_n}.
}
After checking the conditions for Theorem \ref{thm:main} are satisfied we obtain
\eq{
D(P_\theta^n\Vert M^n) \leq {d + 1\over 2} \ln\pi + \ln {1 \over \Gamma({d+1 \over 2})} + {d \over 2} \ln {n \over 2\pi} + o(1).
}
Note that in this case the Bayesian mixture corresponds to the KT estimator for which the redundancy is already
well-known \cite{KT81,Ris84,Wil95,BE06}. It is also worth remarking on the choice of parameterisation. Perhaps the most
natural approach would be to choose $\Theta \subseteq \R^{d+1}$ under the constraint that $\norm{\theta_1} = 1$. Unfortunately
the theorem cannot be applied for a variety of reasons. One is that there is no prior density on
$\Theta \subset \R^{d+1}$ with respect to the $(d+1)$-dimensional Lebesgue measure. Another is that condition (d) is 
no longer satisfied. Even if the theorem held, the bound would
still be worse by an additive ${1 \over 2} \log n$. An overkill solution would be to allow $\Theta$ to be a $d$-dimensional manifold,
but since all conditions are local this is not too helpful in practise. One just ends up working on the 
chart level directly anyway, which essentially is what we do above.

So far we haven't made use of the generalisation to non-stationary dependent sources, but do so now by applying Theorem \ref{thm:main} to 
Markov models.

\subsubsect{Markov models}
One of the simplest commonly used examples of a non-i.i.d.\ process is the Markov chain.
The redundancy of Bayes methods for Markov chains has been studied in detail in \cite{Att99} (see Corollary 1) and \cite{TKB13}. Our
general Theorem \ref{thm:main} implies a
comparable result.
Let $\Omega = \set{1, 2, \cdots, N}$ be a finite set of states. Then a Markov chain
is characterised by a transition matrix $\theta : \R^{N \times N} \to [0,1]$ with $\theta_j^k$ being
the probability of transitioning from state $j$ to state $k$. The vector $\theta_j = \set{\theta_j^1 \cdots, \theta_j^N} \in \R^N$ 
therefore satisfies $\norm{\theta_j} = 1$
for all states $j$. The corresponding measure on the sequence of states is defined in terms of its density by
\eq{
p_{\theta}(\omega_{1:n}) = \prod_{t=1}^n \theta_{\omega_{t-1}}^{\omega_t}
}
where we assume that the intial state $\omega_0 \in \Omega$ is known.
Given the constraint that $\norm{\theta_j} = 1$ allows us to view the parameter space as 
\eq{
\Theta = \set{\theta \in [0,1]^{N(N-1)} : \sum_{k=1}^{N-1} \theta_j^k \leq 1, \forall j}
}
where $\theta_j^N \defined 1 - \sum_{j=1}^{N-1} \theta_j^k$.
The mean Fisher information matrix is no longer independent of $n$, but converges as $n \to \infty$.
If $\pi_\theta(j) = \lim_{n\to\infty} {1 \over n} \sum_{t=1}^n  p_\theta(\omega_t = j|\omega_0)$ is the steady-state distribution of
the Markov chain $p_\theta$, then
the determinant of the Fisher information matrix was shown in \cite{Att99} to be 
\eq{
\det \bi_n(\theta) = \prod_{j=1}^N {\pi_\theta(j)^{N-1} \over \prod_{k=1}^{N} \theta_j^k} 
} 
Therefore if we choose
$w(\theta) \propto \sqrt{\det \bi_n}$, then
\eq{
D(P^n_{\theta} \Vert M^n) \leq {N(N-1) \over 2} \ln {n \over 2\pi} + O(1), 
}
which is known to be the minimax rate.
For detailed discussion on the redundancy of Bayes methods for Markov chains see \cite{Att99,Tak09} and \cite{DMPW81}.

\subsubsect{Discriminative learning}
We now consider a regression setting where a predictor should learn a (noisy) function $f:X \to Y$ from
sequences of data $x_t$ and targets $y_t$.
At each time-step $t$
the predictor observes data $x_t$ in $X$ and should predict a distribution over 
labels $y_t$ in $Y$. The sampled label is then observed and the cycle
repeats. 
It can occur that the sequence of observations is unpredictable and the joint distribution $P(x, y)$ may be hard to model. In this case
we may prefer discriminative learning where $P(y|x)$ is modelled for each $x$. Since $x$ may be arbitrary we 
cannot reasonably assume that the
distribution of the labels $y$ are independent and identically distributed and so the work of Clarke and Barron does not apply.
In order to apply our main theorem we need to be more specific and parameterise $\M = \set{P_\theta(\cdot|x)}$ where
$\theta \in \Theta$. The most natural example is linear regression described below. 

\subsubsect{Linear Basis Function Regression}
Let $x_1, x_2, \cdots, x_n$ be an observed sequence of data where $x_k \in X$ and
$\Phi:X \to [0,1]^d$ be a set of bounded basis function.
Then for $\theta \in \Theta \equiv \R^d$ define a model for the targets $y_{1:n} \in \R^n \equiv \Omega^n$ by
\eq{
p_\theta(y_{1:n}|x_{1:n}) = \prod_{t=1}^n \normal (\theta^\T \Phi(x_t), \beta^{-1} )
}
where the noise parameter $\beta > 0$ is known. The Hessian of $f_n$ in the
statement of Theorem \ref{thm:main-weak} is easily computed to be
\eq{
\partial_{i,j} f_n(\theta) = {\beta \over n} \sum_{t=1}^n \Phi(x_t)_i \Phi(x_t)_j,
}
which is independent of $\theta$. Therefore conditions (b-c) are trivially satisfied.
For condition (a) we simply choose the prior by
\eq{
w(\theta) \defined \normal(\theta|\mu, \Sigma),
}
which is continuous everywhere for all $\mu \in \R^d$ and non-degenerate $\Sigma$.
Other priors are possible, but the Gaussian prior permits efficient computation of the posterior and predictive distributions \cite{Bis07}.
Now 
\eq{
\bi_n(\origin)_{i,j} = \partial_{i,j} f_n(\origin) = {\beta \over n} \sum_{t=1}^n \Phi_i(x_t)\Phi_j(x_t).
}
Since $\Phi(x) \in [0,1]^d$ the spectral norm of the information matrix is bounded by
$\spec{\bi_n(\theta)} \leq d \Vert{\bi_n(\theta)}\Vert_{\max} \leq d \cdot \beta$. Therefore by Theorem \ref{thm:main-weak} we obtain for all $\epsilon > 0$ that
\eq{
&\limsup_{n\to\infty} \left(D_n - {d \over 2} \ln {n \over 2\pi} - {d \over 2} \ln \left({d\cdot \beta} + \epsilon\right) \right)
\leq \ln {1 \over \normal(\origin|\mu,\Sigma)} 
= {1 \over 2} \norm{\origin - \mu}_{\Sigma^{-1}}^2 + {d \over 2} \ln 2\pi + {1 \over 2} \ln \det \Sigma. 
}
Under the assumption that the model is not mis-specificed this shows that Bayesian linear regression 
converges to the truth with low cumulative expected error.

\section{Proof of Theorems \ref{thm:main} and \ref{thm:no-expect}}\label{sec:proof}
As suggested in \cite{Hut05} we roughly follow the proof in \cite{BC90}, carefully adapting each step 
to the non i.i.d.\ case where necessary. 

\begin{lemma}\label{lem:key}
Let $R_n \subseteq \Theta$ be a Lebesgue measurable region with non-zero measure and define a probability measure on $R_n$ by
\eq{
c_n &\defined \int_{R_n} \exp\left(- nf_n(\theta) \right) d\theta \\
\phi_n(\theta) &\defined {1 \over c_n} \exp\left(-nf_n(\theta) \right),
}
Then
\eq{
\tag{$\star$}\ln{p^n_{\origin} \over m^n} &\leq \ln{1 \over w(\origin)} + \ln{1 \over c_n} + \int_{R_n}\!\!\!\phi_n(\theta)\left(\ln{w(\origin) \over w(\theta)}\right) d\theta 
+\int_{R_n}\!\!\! \phi_n(\theta) \left(\ln {p^n_{\origin} \over p^n_\theta} - nf_n(\theta) 
\right) d\theta\\[0.4cm]
\tag{$\star\star$}D_n&\leq \ln{1 \over w(\origin)} + \ln{1 \over c_n} + \int_{R_n}\!\!\! \phi_n(\theta) \left(\ln {w(\origin) \over w(\theta)}\right)d\theta&
}
\end{lemma}
Before presenting the proof we want to remark on one aspect of the difference between our proof and \cite{BC90}. 
They replaced $nf_n(\theta) \equiv D(P^n_{\origin}\Vert P^n_\theta)$ with its Taylor series
in the first step above when defining $\phi_n(\theta)$. With our definition 
we obtain a straight-forward bound on the expected redundancy with almost no assumptions.
The properties of $f_n(\theta)$ can then be used to control the dominant $\ln{1 \over c_n}$ term, with strong assumptions
leading to strong results. Note that the integral term in ($\star\star$) vanishes asymptotically under the assumptions that $w$ is continuous at $\origin$ 
and $R_n$ contracts to the point $\origin$ as $n$ tends to infinity.

\begin{proof}
We follow the proof of Theorem 2.3 in \cite{BC90}, but replace the $\phi_n$ in their proof with the version defined above.
\eq{
&\ln\left({p^n_{\origin} \over m^n}\right) - \ln{1 \over w(\origin)}
\sr{(a)}=  -\ln \left({m^n \over w(\origin) p_{\origin}^n}\right) \\
&\sr{(b)}= -\ln \int_{\Theta} {w(\theta) p_\theta^n \over w(\origin) p_{\origin}^n} d\theta \\
&\sr{(c)}\leq -\ln \int_{R_n} {w(\theta) p_\theta^n \over w(\origin) p_{\origin}^n} d\theta \\
&\sr{(d)}= \ln{1 \over w(\origin)}-\ln \int_{R_n}\!\!\! \phi_n(\theta) {w(\theta) p_\theta^n \over w(\origin) p_{\origin}^n} 
c_n \exp\left(nf_n(\theta)\right) d\theta \\
&\sr{(e)}\leq \int_{R_n}\!\!\! \phi_n(\theta) \ln\left( {w(\origin) p_{\origin}^n \over w(\theta) p_\theta^n} 
{1 \over c_n} \exp\left(-nf_n(\theta)\right)\right) d\theta \\
&\sr{(f)}= \ln{1 \over c_n} +\int_{R_n}\!\!\! \phi_n(\theta)
\left(\ln{p_{\origin}^n \over p_\theta^n} - nf_n(\theta) + \ln\!{w(\origin) \over w(\theta)}  \right) d\theta
}
where (a), (b) and (f) are immediate from definitions and rearrangement.
(c) by positivity of the quantity inside the integral and by restricting the integral to the region $R_n$.
(d) by the introduction of $\phi_n(\theta)$.
(e) is true by Jensen's inequality. The proof of ($\star\star$) follows by 
taking the expectation with respect to $P_{\origin}$ leading to 
\eq{
D_n 
&\equiv \E_{\origin} \ln \left(p^n_{\origin} \over m^n\right)  
\leq  \ln{1 \over w(\origin)} + \ln{1 \over c_n} + \int_{R_n}\!\!\!\! \phi_n(\theta)\left(\ln{w(\origin) \over w(\theta)}\right) d\theta  
}
where we used the fact that $nf_n(\theta) \equiv \E_{\origin} \ln p^n_{\origin} / p^n_\theta$.
\end{proof}
We now prove Theorem \ref{thm:main}. First we choose $R_n$ in such a way that $R_n \to \set{\origin}$ as $n$ tends to infinity.
We then bound $c_n$ by approximating $f_n(\theta)$ using a Taylor series expansion about $\origin$. 
Since $R_n$ contracts
to a point and $w$ is continuous, the term inside the integral in Lemma \ref{lem:key} vanishes.

\begin{proof}[Proof of Theorem \ref{thm:main}]
Fix some $K \in \N$ and define
\eq{
R_n \defined \set{\theta : n\norm{\theta - \origin}_{\bi_n}^2 < K}.
}
By condition (e) in the theorem statement we have that $\spec{\bi_n^{-1}}$ is uniformly bounded by some constant $c$.
Recalling the definition of $R_n$ and Lemma \ref{lem:spec} we obtain
\eq{
\lim_{n\to\infty}& \sup_{\theta \in R_n} n|\theta - \origin|_2^2 
\leq\lim_{n\to\infty} \sup_{\theta \in R_n} n|\theta - \origin|_{\bi_n}^2 \spec{\bi^{-1}_n} 
\leq\lim_{n\to\infty} {K \spec{\bi^{-1}_n}} 
\leq K \cdot c.
}
Therefore $R_n$ contracts to the point-set $\set{\theta_o}$ as $n$ tends to infinity. 
To bound $\ln{1 \over c_n}$ we approximate $f_n(\theta)$ by a second order Taylor series about $\origin$.
It is immediate from the definition that $f_n(\origin) = 0$. The first derivative is the expected 
value of the score function, which vanishes
by assumption (d). 
\eq{
\partial_i f_n(\origin) = 0.
}
The second derivative at $\origin$ is the information matrix $\bi_n$ so
\eq{
f_n(\theta) = {1 \over 2} \norm{\theta - \origin}_{\bi_n}^2 + \norm{\theta - \origin}^2_{h_n(\theta)}
}
where $h_n:\Theta \to \R$ satisfies $\lim_{\theta \to \origin} \sup_n \spec{h_n(\theta)} = 0$ 
by Lemma \ref{A:lem:equicontinuity} in the Appendix. 
We now bound $c_n$ by comparing to the multivariate normal. First we deal with the remainder term. 
\eq{
E_n &\defined \inf_{\theta \in R_n}\exp\left(-n\norm{\theta-\origin}_{h_n(\theta)}^2 \right) \\
&\geq \exp\left(- \sup_{\theta \in R_n} n\spec{h_n(\theta)}^2 \spec{\bi_n^{-1}} \norm{\theta - \origin}_{\bi_n}^2\right) \\
&\geq \exp\left(- \sup_{\theta \in R_n}{\spec{h_n(\theta)}^2 \spec{\bi_n^{-1}} \over K}\right) 
\mathop{\xrightarrow{\hspace{1cm}}}^{n\to\infty}_{R_n \to \set{\origin}}1 
}
where the first inequality follows from Lemma \ref{lem:spec}, the second from the definition of $R_n$. The convergence 
follows from the equicontinuity of the remainder term and the fact that $R_n$ contracts to $\set{\origin}$.
Therefore
\eq{
c_n 
&\equiv \int_{R_n} \exp\left(- nf_n(\theta) \right) d\theta \\
&= \int_{R_n} \exp \left(-{1 \over 2}n |\theta - \origin|_{\bi_n}^2 - n\norm{\theta - \origin}^2_{h_n(\theta)} \right) d\theta \\
&\geq E_n \sqrt{(2 \pi)^d \over \det (n\bi_n)} \int_{R_n} \normal(\theta|\origin, (n\bi_n)^{-1}) d\theta,
}
which is proportional to a normal integral over $R_n$ with mean $\origin$ and variance $(n\bi_n)^{-1}$.
Therefore by a multivariate form of Chebyshev's inequality 
(Lemma \ref{lem:normal} in the Appendix or see \cite{Fer82})
we can bound $c_n$ by 
\eqn{
\nonumber c_n 
&\geq E_n\sqrt{(2\pi)^d \over \det(n\bi_n)} \left(1 - {d / K}\right) \\
\label{eq:c_n-bound}&= E_n\sqrt{\left({2\pi \over n}\right)^d {1 \over \det \bi_n}} \left(1 - {d / K} \right).
}
Therefore by Lemma \ref{lem:key}
\eq{
&D_n \leq \ln{1 \over w(\origin)} + \ln {1 \over 1 - {d / K}} + {d \over 2} \ln {n \over 2\pi} + {1 \over 2} \ln \det \bi_n  
+ \ln{1 \over E_n} + \int_{R_n} \phi_n(\theta) \left(\ln {w(\origin) \over w(\theta)} \right) d\theta.
}
By the continuity of $w$ and the fact that $R_n$ converges to a point
it follows that 
\eq{
\lim_{K \to \infty} \lim_{n\to\infty} \int_{R_n} \phi_n(\theta) \ln {w(\origin) \over w(\theta)} d\theta &= 0. 
}
Therefore sending $K$ to infinity leads to.
\eq{
\limsup_{n\to\infty} \left(D_n - {d \over 2} \ln {n \over 2\pi} - {1 \over 2}\ln \det \bi_n\right)
\leq {1 \over w(\origin)}
}
as required.
\end{proof}

\begin{proof}[Proof of Theorem 2]
We use the same $R_n$ as in the proof of Theorem \ref{thm:main} and apply the first part of Lemma \ref{lem:key} to obtain.
\eqn{
\label{eq:no-expect-bound} 
\ln{p^n_{\origin} \over m^n} &\leq \ln{1 \over w(\origin)} + \ln{1 \over c_n} 
 + \int_{R_n}\!\!\! \phi_n(\theta)\left(\ln {w(\origin) \over w(\theta)} \right)d\theta 
+ \int_{R_n}\!\!\! \phi_n(\theta)\left(\ln {p^n_{\origin} \over p^n_\theta} - nf_n(\theta)\right) d\theta.
}
The $\ln{1 \over c_n}$ term and the first integral are constant random variables and 
are bounded as they were in Theorem \ref{thm:main}.
We only need to show that
\eq{
r_n \defined\int_{R_n}\!\!\! \phi_n(\theta)\left(\ln {p_{\origin}^n \over p_\theta^n} - nf_n(\theta)\right)d\theta
}
converges to zero in $L^1$.
By an application of Fubini's theorem the expectations can be 
exchanged to yield
\eq{
\E_{\origin}|r_n| 
&= \int_{R_n}\!\!\! \phi_n(\theta)\E_{\origin}\left|\ln{p_{\origin}^n \over p_\theta^n} - 
nf_n(\theta)\right| d\theta \\
&\leq \int_{R_n}\!\!\! \phi_n(\theta)\left(\Var_{\origin} \ln{p_{\origin}^n \over p_\theta^n}\right)^{{1 \over 2}} d\theta.
}
We now use assumption (f) to control the variance term by taking a second degree Taylor series 
expansion about $\origin$. If $\theta \in R_n$, then 
\eq{
\Var_{\origin} \left(\ln {p^n_{\origin} \over p^n_\theta}\right) 
&= n\norm{\theta - \origin}^2_{\bi_n} + n\norm{\theta - \origin}_{h_n(\theta)}^2 \\
&\leq n\norm{\theta - \origin}_{\bi_n}\left(1 +  \spec{\bi_n^{-1}} \spec{h_n(\theta)}^2\right) \\
&\leq K (1 + \spec{\bi_n^{-1}}\spec{h_n(\theta)}^2)
}
where the second derivative of the variance at $\origin$ is twice the information matrix with remainder $h_n:\Theta \to \R$ satisfying.
$\lim_{\theta \to \origin} \sup_n \spec{h_n(\theta)} = 0$ by Lemma \ref{A:lem:equicontinuity}. In the 
last two steps we used Lemma \ref{lem:spec}
and the definition of $R_n$.
Therefore
\eqn{
\nonumber 
\lim_{n\to\infty} \E_{\origin}|r_n| 
&\leq 
\lim_{n\to\infty}\int_{R_n} \!\!\!\phi_n(\theta) \left(\Var_{\origin} \ln{ p^n_{\origin} \over p^n_\theta} \right)^{1 \over 2} d\theta \\
\label{eq:Var-bound} & \leq \lim_{n\to\infty} \int_{R_n} \!\!\! \phi_n(\theta) \left(K(1 + \spec{\bi_n^{-1}}\spec{h_n(\theta)}^2)\right)^{1 \over 2} d \theta \\ 
&= \sqrt{K} 
}
where the final inequality follows since $R_n$ contracts to a point and because $\spec{\bi_n^{-1}}$ is uniformly bounded.
Therefore $\lim_{n\to\infty} \E_{\origin} |r_n| = 0$. We cannot send $K$ to infinity like in the proof of Theorem \ref{thm:main}.
Instead we simply fix $K = 2d$ and insert Equations (\ref{eq:c_n-bound}) and (\ref{eq:Var-bound}) into (\ref{eq:no-expect-bound}) to
obtain for all $\omega = \omega_1 \omega_2 \cdots \in \Omega^\infty$
\eq{
\limsup_{n\to\infty} &\left( \ln {p_{\origin}^n(\omega_{1:n}) \over m^n(\omega_{1:n})} - {d \over 2} \ln {n \over 2\pi} - {1 \over 2} \ln \det \bi_n - r_n(\omega)\right) 
\leq \ln{1 \over w(\origin)} + \ln 2 + \sqrt{2d}
}
as required. 
\end{proof}

\section{Proof of Theorem \ref{thm:main-weak}}\label{sec:proof-weak}

In the previous section we used the assumption (e) that $\spec{\bi_n^{-1}}$ was uniformly bounded by a constant
to show that the critical region
$R_n$ contracts to $\origin$ as $n$ tends to infinity. This result can be guaranteed by defining the region $R_n$ with
respect to a different norm.
Define a positive definite matrix $A_n \defined \bi_n + \epsilon I$ with $I \in \R^{d\times d}$ the identity matrix and $\epsilon > 0$ to be chosen later.
The matrix $A_n$ is positive definite since it is the sum of two positive definite matrices. Furthermore $A_n > \bi_n$.
The assumption may now be eliminated by using the norm $|\theta - \origin|_{A_n}$ to define the critical region
rather than $|\theta - \origin|_{\bi_n}$ . The only component that requires checking is the
bound on $c_n$.
The critical region becomes
\eq{
R_n \defined \set{n \norm{\theta - \origin}_{A_n}^2 \leq K}.
}
We now show that $R_n$ contracts to a point as 
$n$ tends to infinity. It is an easy consequence of the spectral decomposition theorem that the
smallest eigenvalue of $A_n$ satisfies $\lambda_{\min} \geq \epsilon$
and $\spec{A_n^{-1}} = {1 \over \lambda_{\min}} \leq {1 \over \epsilon}$, which is precisely the condition required for $R_n$ to contract to the
point $\origin$.
Therefore
\eq{
\limsup_{n\to\infty} \left(D_n - {d \over 2} \ln {n \over 2\pi} - {1 \over 2} \ln\det A_n \right) \leq \ln{1 \over w(\origin)}
}
The determinant of $A_n$ may asymptotically be arbitrarily larger than $\det \bi_n$, which appeared in Theorem \ref{thm:main}.
If $\lambda_i$ is the $i$th eigenvalue of $\bi_n$, then $\lambda_i + \epsilon$ is the $i$th eigenvalue of $A_n$ and an easy bound on $\det A_n$ is
\eq{
\det A_n &= \prod_{i=1}^d (\lambda_i + \epsilon) 
\sr{(a)}\leq \prod_{i=1}^d (\spec{\bi_n} + \epsilon) 
= (\spec{\bi_n} + \epsilon)^d
}
where (a) follows by the inequality $\lambda_i \leq \spec{\bi_n}$.
Therefore
\eq{
\limsup_{n\to\infty} 
\left(D_n - {d \over 2} \ln{n \over 2\pi} - {d \over 2} \ln(\spec{\bi_n} + \epsilon) \right) \leq \ln{1 \over w(\origin)}
}
as required. \hfill \qedsymbol

\section{Weakening the assumptions}\label{sec:equi}
The continuity of the prior $w$ at $\origin$ is required to ensure that $w$ assigns non-zero probability to environments
in a region of $\origin$. That $f_n$ is twice differentiable is required to define the curvature of $\E_{\origin} \ln p_\theta$. 
Condition (d) was used to show that the 1st order terms in the Taylor series of $f_n$ vanish. The assumption is standard and
very weak. For example, it holds for all finite $\Omega$, as well as exponential families in their canonical form.
A nice discussion with examples may be found in \cite[\S18]{Gru07}.
The equicontinuity condition
is necessary to ensure that the Taylor approximation of $f_n$ is uniformly accurate in $n$. Actually a counter-example when 
$\partial_{i,j} f_n$ is not equicontinuous
is not hard to construct.
Let $\Omega = \set{0,1}$, $\Theta = [0,1]$ and $w(\theta) = 1$ and $a_n \geq 0$ be a sequence of constants to be chosen later. 
Define probability mass in terms of its conditionals by
\eq{
p_\theta(1|\omega_{<n}) 
\defined \min\set{1,\theta + a_n \left(\theta - {1 \over 2}\right)^2},
}
which depends on $n$, but not $\omega_{<n}$.
For $\origin = {1 \over 2}$ we have $p_{\origin}(1|\omega_{<n}) = {1 \over 2}$ for all $n$, but 
if $a_n$ is chosen to converge to infinity, then $\lim_{n\to\infty} p_\theta(1|\omega_{<n}) = 1$ for all $\theta \neq \origin$. 
Based on this, if $a_n$ is chosen to converge to infinity sufficiently fast, 
then width of the interval in which $p_\theta^n$ is close to $p_{\origin}^n$ can be made arbitrarily small (Fig. \ref{fig:plot}) and so
\eq{
\lim_{n\to\infty} \left(D_n - {1 \over 2} \ln {n \over 2\pi}\right) = \infty. 
}

\begin{figure}
\centering
\includegraphics{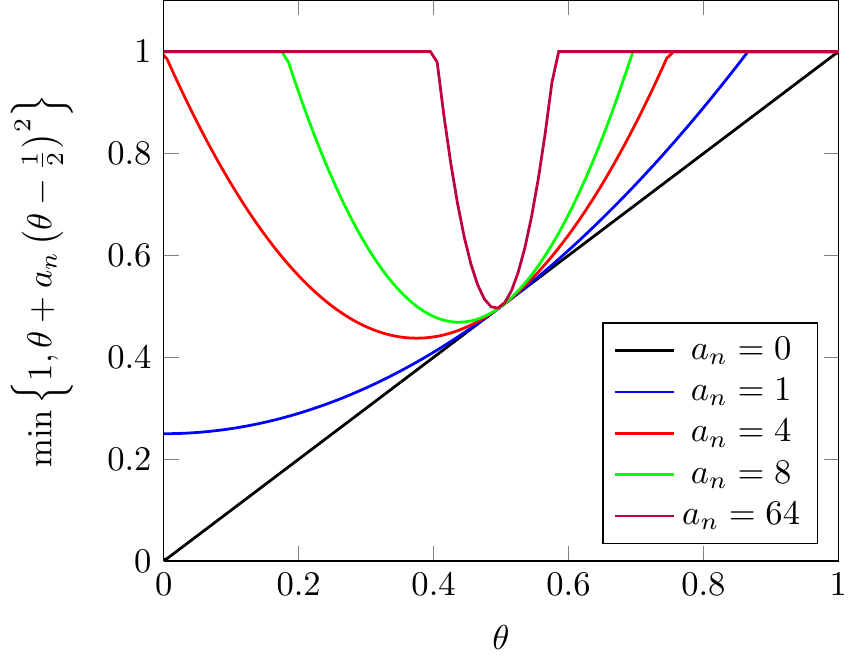}
\caption{Conditional probability mass $p_\theta(1|\omega_{<n})$ for different $a_n$}
\label{fig:plot}
\end{figure}
The information matrix at $\origin = {1 \over 2}$ can be found via a straight-forward application of the definition and 
is $\bi_n(\origin) = 1/4$, which is positive (definite) and independent of $n$. 
Therefore all conditions of the theorem are met except for equicontinuity and yet the result does not hold.
This should not be surprising. The equicontinuity condition ensures that the approximation of $f_n$ via its second
order Taylor series has uniform error, which is crucial to the proof. 

\section{Higher Order Derivatives}\label{sec:higher-order}

In the main theorems we used the second order approximation of the log likelihood function to control the
redundancy, but if $\bi_n = 0$, then Theorem \ref{thm:main} cannot be applied and Theorem \ref{thm:main-weak} seems 
suboptimal since, regardless of how $\epsilon$ is chosen, the bound on $D_n$ still increases like ${d \over 2} \log n$.
Here we consider the case where the second (and maybe higher) derivatives 
vanish. 
For the remainder of this section we fix an even $k$ and assume that $\partial_{ i_1 \cdots  i_j}f_n(\origin) = 0$ for all $j < k$.
Then by the $k$th order Taylor expansion we let
\eq{
f_n(\theta) \approx {1 \over k!} \sum_{ i_1 \cdots  i_k} \partial_{ i_1 \cdots  i_k} f_n(\origin) (\theta - \origin)^ i 
}
where for $x \in \R^d$ and $ i =  i_1 \cdots  i_k$ we use the multi-index notation
$x^ i \equiv \prod_{i=1}^k x_{ i_i}$.
Then define 
\eq{
\Lambda_n \defined \max_{\substack{x \in \R^d \\ \norm{x}_k = 1}} 
\sum_{ i_1 \cdots  i_k} \partial_{ i_1\cdots  i_k} f_n(\origin)x^ i 
}
where $\Lambda_n \in (0, \infty)$ because the derivative is positive definite and the set $\set{x \in \R^d : \norm{x}_k = 1}$ is compact.
Then
\eq{
\sum_{ i_1\cdots i_k} \partial_{ i_1\cdots i_k} f_n(\theta_o)(\theta - \origin)^ i
\leq \Lambda_n \sum_{i=1}^d (\theta_i - {\origin}_i)^k.
}
We replace the critical region used in the proof of Theorem \ref{thm:main} with
\eq{
R_n \defined \set{\theta : n \Lambda_n \sum_{i=1}^d (\theta_i - {\origin}_i)^k < K}. 
}
The proof goes through unchanged except for the computation of $c_n$.
We define a measure on $\R^d$ by
\eq{
G_n(\theta) \defined {1 \over \eta_{n}} \exp\left(-{n\Lambda_n \over k!} \sum_{i=1}^d (\theta_i - {\origin}_i)^k\right)
}
where the $\eta_{n}$ is the normalisation constant that can be determined by the usual methods for computing Gaussian integrals.
\eq{
\eta_{n} 
&\defined \int_{\R^d} \exp\left(-{n\Lambda_n \over k!}\sum_{i=1}^d (\theta_i - {\origin}_i)^k \right) d\theta  
=\left(\int_{-\infty}^\infty \exp\left(-{n \Lambda_n \over k!} x^{k}\right) dx \right)^d 
=\left({2 \over k}\left({k! \over n \Lambda_n}\right)^{1/k} \Gamma\left({1 \over k}\right)\right)^d.
}
where the last step follows by substituting $y = n\Lambda_n x^k/k!$ and the definition of the Gamma function.
Then
\eq{
c_n 
&= \int_{R_n} \exp\left(-n f_n(\theta)\right) d\theta 
\approx \eta_n \int_{R_n} G_n(\theta) d\theta
}
where the approximation is due to the lower order terms when substituting the Taylor series expansion. 
The integral may be bounded naively by computing the $k$th moments of $G$ in each direction, Markov's inequality and the union bound.
We omit the details. This approach leads to the
bound
\eq{
\int_{R_n} G(\theta) d\theta \geq 1 - {d^2(k-1)! \over K} \sr{K \to \infty}\longrightarrow 1,
}
which importantly is independent of $n$. Like in the proof of Theorem \ref{thm:main} we can send $K$ to infinity in the final stage of the
proof.
\eq{
c_n \geq \eta_n \left(1 - {d^2(k - 1)! \over K}\right)
}
Finally we apply Lemma \ref{lem:key} to obtain $\limsup_{n\to\infty} \left(D_n - \ln{1 \over \eta_n}\right) \leq \ln{1 \over w(\origin)}$ and so
\eq{
&\limsup_{n\to\infty} \left(D_n - {d \over k} \ln n - {d \over k} \ln \Lambda_n\right) 
\leq \ln{1 \over w(\origin)} + {d \over k} \ln {1 \over k!} + d \ln {k \over 2\Gamma({1 \over k})}.
}
Most interesting is the dependence on $n$, which if $\Lambda_n$ is assumed to be constant in $n$
is logarithmic, regardless of $k$. Larger $k$ only decrease the multiplicative constant.
Note that for $k = 2$ we have ${d \over 2} \ln \Lambda_n = {d \over 2}\ln \spec{\bi_n} \geq {1 \over 2} \ln \det{\bi_n}$,
which featured in Theorem \ref{thm:main}. An obscure case is when the $k$th derivatives of $f_n(\theta)$ vanish for all
$k$, but where $f_n(\theta)$
is itself non-zero. This can occur, for example if $\Omega = \set{0,1}$, $\Theta = [0,1]$, $w(\theta) = 1$ and $p_\theta^n(\omega_{1:n})$ is
\iid{} with $p^1_\theta(0) = \exp(-\exp(-1/\theta))$ where $p^1_0(0) \defined 1$. In this case $D_n$ can be shown
to grow with order $\log \log n$.

\section{Zero Dimensional Families}\label{sec:zero}

All results have been proven for $d \geq 1$. Here we briefly compare to the case when 
$\M = \set{P_1, P_2, \cdots}$ is a countable family of measures. It is no longer possible, desirable or 
necessary to define the Bayes mixture in terms of densities or a continuous prior.
Instead $w:\N \to [0,1]$ is a probability mass function and the Bayes mixture is defined by
\eq{
M(A) \defined \sum_{i=1}^\infty w(i) P_i(A).
}
Then it is trivial to bound $D_n = D(P^n_k \Vert M^n) \leq \ln {1 \over w(k)}$, which is independent of $n$. 
$\M$ may be parameterised by $\Theta = \R^+$ where $P_{\theta} = P_k$ for $\theta \in [k-1,k)$. With this
parameterisation the log likelihood is constant in a region about 
$\origin = k-{1 \over 2} \in (k-1,k)$, but in this case even the bound
shown in the previous section guarantees only sub-logarithmic redundancy, when in fact it should be constant.

\section{Conclusion}\label{sec:conc}

Our main contribution is a generalisation of Theorem 2.3 in \cite{BC90} to the case where sources are permitted to be dependent and non-stationary.
Under only mild assumptions we obtain the same bound on the relative entropy between the Bayesian mixture and $P_{\theta} \in \M$ 
of order ${d \over 2} \log n$ where $d$ is the dimension of the hypothesis class $\M$. The results were applied to discriminative learning and Markov chains
for which the results in \cite{BC90} do not apply. 
If the mean Fisher information has unbounded spectral norm as $n$ tends to infinity, then
Theorem \ref{thm:main-weak} can be applied to obtain nearly the same bound as Theorem \ref{thm:main}.
We showed that if the information matrix vanishes, then higher-order approximations of the log-likelihood lead to
similar bounds on the redundancy. Our results can usefully be applied beyond discriminative learning and Markov sources
presented here and suggest 
Bayesian reinforcement learning as a natural and important example of non-\iid sources \cite{Hut05}.
Interesting future work is to generalise the other results in \cite{BC90}.

\section*{Acknowledgment}
Tor Lattimore was supported by the Australian Google Fellowship for Machine Learning.

\bibliography{all}

\appendix

\section{Technical Results}

\begin{lemma}[\cite{Fer82}]\label{lem:normal} 
Suppose $\theta \sim \normal(\origin, \Sigma)$ where $\origin \in \R^d$ and $\Sigma \in \R^{d \times d}$ is positive definite.
Then for all $\delta > 0$
\eq{
P(|\theta - \origin|_{\Sigma^{-1}}^2 \leq \delta) \geq {1 - {d \over \delta}}.
}
\end{lemma}

The following lemma is required to uniformly bound the approximation error of the Taylor series of $\ln {P^n_{\theta} / P^n_{\origin}}$. 
We presume similar results appear elsewhere, but include the statement and proof for completeness and because references seem hard to find.

\begin{lemma}\label{A:lem:equicontinuity}
Let $\origin \in \Theta \subset \R^d$ and $f_n:\Theta \to \R$ and $H_n(\theta) \in \R^{d\times d}$ be the Hessian
of $f_n$ at $\theta$. Suppose also that 
\begin{enumerate}[(a)]
\item $f_n$ is twice differentiable at $\origin \in \Theta$.
\item $\partial_i f_n(\origin) = 0$ and $f_n(\origin) = 0$ for all $1 \leq i \leq d$.
\item $H_n(\theta)_{i,j}$ is equicontinuous at $\origin$ for all $1 \leq i,j \leq d$.
\end{enumerate}
Then there exists a family of functions $h_n:\Theta \to \R^{d\times d}$ such that
\eq{
f_n(\theta) &= {1 \over 2} \norm{\theta - \origin}_{H_n(\origin)}^2 + \norm{\theta - \origin}_{h_n(\theta)}^2
}
where $\lim_{\delta \to 0} \sup_n \sup_{\theta \in N_\delta^2} \spec{h_n(\theta)} = 0$
and $N_\delta^2 \defined \set{\theta : \norm{\theta - \origin}_2 < \delta}$.
\end{lemma}

\begin{proof}
For $\theta \in N_\delta^2$ a first order Taylor series of $f_n$ about $\origin$ leads to
\eq{
f_n(\theta) = (\theta - \origin)^\T S_n(\theta)(\theta - \origin)
}
where 
\eq{
S_n(\theta)_{i,j} \in \set{ {1 \over 2} H_n(\bar\theta)_{i,j} : \bar\theta \in N_\delta^2}
}
Now define
$h_n(\theta) \defined S_n(\theta) - {1 \over 2} H_n(\theta_\circ)$, which
implies that $f_n(\theta) = {1 \over 2}|\theta - \origin|_{H_n} + |\theta - \origin|_{h_n(\theta)}$ and
\eq{
\lim_{\delta \to 0} &\sup_{n} \sup_{\theta \in N_\delta^2} \spec{h_n(\theta)}
\sr{(a)}\leq d \lim_{\delta \to 0} \sup_n \sup_{\theta \in N_\delta^2} \norm{h_n(\theta)}_{\max} \\ 
&\sr{(b)}= d \lim_{\delta \to 0} \sup_n \sup_{\theta \in N_\delta^2} 
\max_{i,j} \left|S_n(\theta)_{i,j} - {1 \over 2} H_n(\origin)_{i,j}\right| \\
&\sr{(c)}\leq d \max_{i,j} \lim_{\delta\to 0} \sup_n \sup_{\theta \in N_\delta^2} \left|{1 \over 2} H_n(\theta)_{i,j} - {1 \over 2}H_n(\origin)_{i,j}\right|\\
&\sr{(d)}= 0
}
where (a) is follows from the bound $\spec{\cdot} \leq d \norm{\cdot}_{\max}$ for matrices of dimension $d$.
(b) by the definition of $h_n(\theta)$.
(c) by the definition of $S_n(\theta)$ and exchanging the max with the limit, which 
is valid because the indices $i$ and $j$ range over finite set.
(d) by equicontinuity of the Hessian at $\theta_o$.
\end{proof}

\begin{lemma}\label{lem:spec}
If $A \in \R^{d\times d}$ is positive definite and $x \in \R^d$, then $\norm{x}_2^2 \leq \norm{x}_A^2 \spec{A^{-1}}$.
If $B \in \R^{d\times d}$ and $x \in \R^d$, then $\norm{x}_B^2 \leq \norm{x}_2^2 \spec{B}^2$.
\end{lemma}

\begin{proof}
Let $\lambda_{\min}$ be the smallest eigenvalue of $A$ and $A = U^{\T} D U$ be the spectral decomposition of $A$ where the diagonal of $D$ consists of the eigenvalues of $A$ and $U^{-1} = U^{\T}$.
Then $\norm{x}_2^2 = x^{\T} U^{\T} Ux \leq x^{\T}U^{\T} D U x / \lambda_{\min} = \norm{x}_A^2 / \lambda_{\min}$. 
Complete the first part by checking that $\spec{A^{-1}} = 1/\lambda_{\min}$.
The second part is proven by
noting that
\eq{
{\norm{x}_B^2 \over \norm{x}_2^2} \leq \max_{y:\norm{y}_2^2 = \norm{x}_2^2} {\norm{y}_B^2 \over \norm{y}_2^2} \equiv \spec{B}^2.
}
Rearrange to complete the result.
\end{proof}

\begin{lemma}\label{lem:mat}
Let $\sum_{j=0}^{d} \theta_j = 1$.
Then the matrix $A \in \R^{d\times d}$ defined by
\eq{
A_{j,k} = {\ind{j=k} \over \theta_j} + {1 \over \theta_{0}}
}
has determinant $\det A = \prod_{j=0}^{d} {1 \over \theta_j}$.
\end{lemma}

\begin{proof}
Define $\alpha_j = 1/\theta_j$
and subtract the last row of $A$ from all previous rows to obtain a block matrix $B$ with the same determinant as $A$ of the form
\eq{
B &= 
\left(\begin{matrix}
W & X \\
Y & Z
\end{matrix}\right)  
=
\left(\begin{array}{cccc|c}
\alpha_1  & 0         & \cdots & 0            & -\alpha_d \\
0         & \alpha_2  &       &              & -\alpha_d  \\
\vdots    &           & \ddots &              & -\alpha_d\\
0         & \cdots           & 0       & \alpha_{d-1} & -\alpha_d\\ \hline
\alpha_{0}         & \alpha_{0}         & \alpha_{0}      & \alpha_{0}            & \alpha_d + \alpha_{0}
\end{array}\right)
}
Then 
\eq{
\det A &= \det W \det (Z - YW^{-1} X) \\
&= \left(\prod_{j=1}^{d-1} \alpha_j\right) \left(\alpha_d + \alpha_{0} + \sum_{j=1}^{d-1} {\alpha_d \alpha_{0} \over \alpha_j}\right)\\
&= \left(\prod_{j=0}^{d} \alpha_j\right) \left({1 \over \alpha_{0}} + {1 \over \alpha_{d}} + \sum_{j=1}^{d-1} {1 \over \alpha_j}\right)\\
&= \left(\prod_{j=0}^{d} \alpha_j\right) \left(\sum_{i=0}^{d} \theta_i\right) 
= \prod_{j=0}^{d} \alpha_j 
\equiv \prod_{j=0}^{d}{1 \over \theta_j}
}
as required.
\end{proof}

\end{document}